\documentclass[a4paper,USenglish]{lipics}

\usepackage{graphicx}
\usepackage{stmaryrd}
\usepackage{amsmath} 

\newcommand{\CA}{\operatorname{CA}} 
\newcommand{\Lprop}[1]{L^{[#1]}} 
\newcommand{\ACA}{\mathcal{A}} 
\newcommand{\Conf}{\mathfrak{C}} 
\newcommand{\Neigh}{\mathcal{N}} 
\newcommand{\Configs}{\operatorname{Conf}} 
\newcommand{\NN}{\mathbb{N}} 
\newcommand{\ZZ}{\mathbb{Z}} 
\newcommand{\QQ}{\mathbb{Q}} 

\newcommand{\stB}{\raisebox{-.2em}{\includegraphics[page=4,scale=.8]{symbols.pdf}}}

\newcommand{\stE}{\raisebox{-.2em}{\includegraphics[page=2,scale=.8]{symbols.pdf}}}

\title{Comparing 1D and 2D Real Time on Cellular Automata}
\author[1]{Anaël Grandjean}
\author[1]{Victor Poupet}
\affil[1]{LIRMM, Université Montpellier 2\\
161 rue Ada, 34392 Montpellier, France}
\authorrunning{A. Grandjean and V. Poupet}
\Copyright{Anaël Granjean and Victor Poupet}

\serieslogo{}
\volumeinfo
  {Billy Editor and Bill Editors}
  {2}
  {Conference title on which this volume is based on}
  {1}
  {1}
  {1}
\EventShortName{}
\DOI{10.4230/LIPIcs.xxx.yyy.p}

\begin{document}

\maketitle

\begin{abstract}
	We study the influence of the dimension of cellular automata (CA) for real time language recognition of one-dimensional languages with parallel input. Specifically, we focus on the question of determining whether every language that can be recognized in real time on a 2-dimensional CA working on the Moore neighborhood can also be recognized in real time by a 1-dimensional CA working on the standard two-way neighborhood.
	
	We show that 2-dimensional CA in real time can perform a linear number of simulations of a 1-dimensional real time CA. If the two classes are equal then the number of simulated instances can be polynomial.
\end{abstract}

\section{Introduction}

Cellular automata (CA) were first introduced in the 1940s (published posthumously in 1966) by J.~von Neumann and S.~Ulam as a mathematical model to study self-replication \cite{neumann66}. Although initially studied as a a dynamical system, A.R.~Smith III proved that it was possible to embed Turing machines in their behavior \cite{smithIII71} and were as such a convenient model for massively parallel computation.

Cellular automata are also well suited to work on various dimensions. The original CA by von Neumann is 2-dimensional, but the natural simulation of Turing machines is on one dimension. CA represent therefore a natural way to study how the dimension of the space affects the computing power of the machines \cite{cole69,terrier06}.

This article presents some results comparing the computational power of 1-dimensional and 2-dimensional CA on parallel input. The main open question in that respect is to determine whether or not all languages that can be recognized in real time on 2-dimensional CA can also be recognized in real time in 1 dimension \cite{delorme00}.

The organization of the article is as follows. Section~\ref{sec:definitions} recalls the basic definitions and concepts about cellular automata that are used throughout the article. Section~\ref{sec:markers} presents a construction on 1-dimensional CA that shows how it is possible to consider that a real time CA knows approximately where the middle (or any other fixed rational proportion) of the input word is from the start. This construction is used to prove the main theorem of Section~\ref{sec:power of space}. In Section~\ref{sec:central compression} we present a classic technique on cellular automata that compresses the space-time diagram of a 1-dimensional CA. The main novelty here is that we perform the compression not on the middle of the input word but on an approximate position ``near the center''. Section~\ref{sec:power of space} presents and proves the main result of the article, which states in essence that 2-dimensional CA can simulate in real time a linear number of simulations of a 1-dimensional real time CA. Finally Section~\ref{sec:consequences} discusses some consequences of the main theorems.

\section{Definitions}
\label{sec:definitions}

\subsection{Cellular Automata}

\begin{definition}[Cellular Automaton]
	A \emph{cellular automaton} (CA) is a quadruple $\ACA=(d, Q, \Neigh, \delta)$ where:
	\begin{itemize}
		\item $d\in\NN$ is the dimension of $\ACA$;
		\item $Q$ is a finite set whose elements are called \emph{states};
		\item $\Neigh\subset \ZZ^d$ is a finite set called \emph{neighborhood} of $\ACA$ such that $0\in\Neigh$;
		\item $\delta:Q^\Neigh \rightarrow Q$ is the local transition function of $\ACA$.
	\end{itemize}
	
	A \emph{configuration} of the automaton is a mapping $\Conf:\ZZ^d\rightarrow Q$. The elements of $\ZZ^d$ are called \emph{cells} and for a given cell $c\in\ZZ^d$, we say that $\Conf(c)$ is the state of $c$ in the configuration $\Conf$. The set of all configurations over $Q$ is denoted $\Configs(Q)$. For a given configuration $\Conf\in\Configs(Q)$ and a cell $c\in\ZZ^d$, define the \emph{neighborhood of $c$ in $\Conf$}
	\begin{displaymath}
		\Neigh_\Conf(c) = \left\{
			\begin{array}{rcl}
				\Neigh & \rightarrow & Q \\
				n & \mapsto & \Conf(c+n)
			\end{array}
		\right.
	\end{displaymath}
	From the local transition function $\delta$, we define the \emph{global transition function} $\Delta_\ACA$ of the automaton. The image of a configuration $\Conf$ by $\Delta_\ACA$ is obtained by replacing the state of each cell $c$ by the image by $\delta$ of the neighborhood of $c$ in $\Conf$ :
	\begin{displaymath}
		\Delta_\ACA: \left\{
			\begin{array}{rcl}
				\Configs(Q) & \rightarrow & \Configs(Q) \\
				\Conf & \mapsto & \left\{
					\begin{array}{rcl}
						\ZZ^d & \rightarrow & Q \\
						c & \mapsto & \delta(\Neigh_\Conf(c))
					\end{array}
				\right.
			\end{array}
		\right.
	\end{displaymath}
\end{definition}

In this article, we will only consider 1-dimensional CA working on the standard neighborhood
$\Neigh_{\operatorname{std}}=\{-1,0,1\}$
and 2-dimensional cellular automata working on the Moore neighborhood 
$\Neigh_{\operatorname{M}} = \{(x,y)\ |\ -1\leq x, y\leq 1\}$
(see Figure~\ref{fig:neighborhoods}).

\begin{figure}[htbp]
	\centering
		\includegraphics[scale=1]{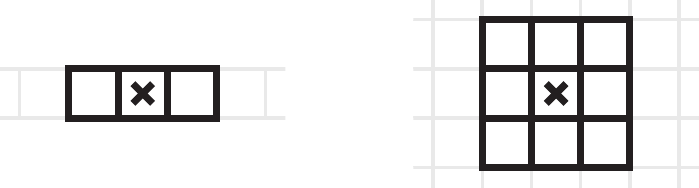}
	\caption{The standard 1-dimensional neighborhood (left) and the Moore 2-dimensional neighborhood (right).}
	\label{fig:neighborhoods}
\end{figure}

\subsection{Language Recognition}

\begin{definition}[Language Recognizer]
	Given a finite alphabet $\Sigma$ and a language $L\subseteq \Sigma^*$, a $d$-dimensional CA $\ACA$ with states $Q$ is said to recognize $L$ in time $f:\NN \rightarrow \NN$ with accepting states $Q_a\subseteq Q$ and quiescent state $q_0\in Q$ if, $\Sigma\subseteq Q$ and for any word $w=u_0u_1\ldots u_{n-1}\in\Sigma^*$, starting from the configuration
\begin{displaymath}
	\begin{array}{rcl}
		\ZZ^d & \rightarrow & Q \\
		(x, y_1, y_2, \ldots, y_{d-1}) & \mapsto & \left\{ \begin{array}{l}
			u_{x} \quad\textrm{if}\quad x\in\llbracket 0, n-1\rrbracket \quad\textrm{and}\quad \forall i, y_i=0\\
			q_0 \quad\textrm{ otherwise}
		\end{array}\right.
	\end{array}
\end{displaymath}
the state of the origin at time $f(n)$ is in $Q_a$ if and only if $w\in L$.
\end{definition}

\begin{definition}[Real and Linear Time]
	The \emph{real time} function is the function $n\mapsto n-1$. This time function corresponds to the minimal time necessary for information held on the last letter of the input word to reach the origin and hence affect the recognition of the word. The class of languages recognized in real time on 1 dimensional (resp. 2-dimensional) CA will be denoted $\CA(n)$ (resp. $\CA_2(n)$).
	
	We will say that a language is recognized in \emph{linear time} if it can be recognized in time $n\mapsto 2n$. The class of languages recognized in linear time on 1-dimensional (resp. 2 dimensional) CA will be denoted $\CA(2n)$ (resp. $\CA_2(2n)$).
\end{definition}

Because there are linear acceleration theorems on 1-dimensional and 2-dimensional CA \cite{mazoyer92} (on the simple neighborhoods that we consider), any language recognized in time $n\mapsto kn$ for $k > 0$ is also recognized in time $n\mapsto 2n$, which explains the denomination of \emph{linear time}. As for real time, it is a long open question to determine whether $\CA(n)=\CA(2n)$.

In this article, we investigate whether adding a dimension to the automaton increases linear and real time recognition power, namely if $\CA(n)=\CA_2(n)$ and $\CA(2n)=\CA_2(2n)$. These questions are long open problems (see problem 26 in \cite{delorme00}).

\subsection{Tools}

\subsubsection{Space-Time Diagram}
A space-time diagram is a 2-dimensional representation of the evolution of a 1-dimensional CA from a specific configuration. Each configuration in the evolution is represented by a line of the diagram, with time going from bottom to top. We do not usually consider space-time diagrams of 2-dimensional CA as these would be in 3 dimensions.

A specific point in space and time will be referred to as a \emph{site} of the space-time diagram.

\subsubsection{Layers}
Given a CA $\ACA=(d,Q,\Neigh,\delta)$, \emph{adding a layer} to $\ACA$ that performs a certain task consists in designing a specific CA working on a set of states $Q'$ that performs the task and extending the set of states of $\ACA$ to the product $Q\times Q'$. In doing so, the new product automaton can mimic the behavior of $\ACA$ on its first coordinate and perform the new task on the second coordinate. From there, it is possible to modify the behavior of the automaton by having the two layers interact with each other.

As long as each layer requires only a finite number of states and there are only a finite number of layers, the total number of states of the resulting automaton remains finite.

\section{Markers}
\label{sec:markers}

In this section we investigate whether ``marking'' specific positions on the input word can help real time recognition of a language. The results in this section are a generalization of a technique used by O.~Ibarra and T.~Jiang in their proof that if $\CA(n)$ is closed under reversal then $\CA(n)=\CA(2n)$ \cite{ibarra88}.

Given a finite alphabet $\Sigma$, we mark some letters of words of $\Sigma^*$ by considering the extended alphabet $\Sigma\times \{0,1\}$. We say that the word $(u_i,\delta_i)_{i\in\llbracket 0, n-1\rrbracket} \in (\Sigma\times\{0,1\})^*$ corresponds to the word $(u_i)_{i\in \llbracket 0, n-1\rrbracket}$ where all the positions $i$ such that $\delta_i=1$ have been marked ($\delta_i=0$ means that the letter has not been marked).

\subsection{Exact and Fuzzy Marking}

\begin{definition}[Proportional Marking]
	Given $\alpha \in [0, 1[$ and a language $L$ over an alphabet $\Sigma$, we define $\Lprop\alpha\in(\Sigma\times\{0,1\})^*$ as the language of words of $L$ for which only the letter at position $\lfloor\alpha n\rfloor$ has been marked ($n$ is the length of the word).
	
	Formally, the word $(u_i,\delta_i)_{i\in\llbracket 0, n-1\rrbracket} \in (\Sigma\times\{0,1\})^*$ is in $\Lprop\alpha$ if and only if $u_0u_1\ldots u_{n-1}\in L$, $\delta_{\lfloor\alpha n\rfloor} = 1$ and for all other $i$, $\delta_i = 0$.
\end{definition}

When working on real time CA algorithms it would sometimes be convenient to know where the middle of the word is, or some other specific ratio. Whether marking the letter of the input word corresponding to a fixed proportion of the word length can help recognize in real time languages that were not in $\CA(n)$ is still an open question to our knowledge\footnote{Note that if $\CA(n) = \CA(2n)$ then proportional marking with a rational ratio does not help real time recognition since the cell at position $\lfloor\alpha n\rfloor$ can be marked in $n$ time steps for any rational $\alpha$.}. Although an answer to this question would be very interesting, we can actually make many constructions with a weaker version that we can prove : instead of requiring a mark on the exact cell at position $\lfloor\alpha n\rfloor$, it is enough to have a mark on one of the cells between positions $\lfloor\alpha n\rfloor$ and $\lfloor\beta n\rfloor$ for $\alpha < \beta$.

\begin{definition}[Fuzzy Marking]
	Given $\alpha, \beta \in [0, 1[$ with $\alpha\leq \beta$ and a language $L$ over an alphabet $\Sigma$, we define $\Lprop{\alpha,\beta}\in(\Sigma\times\{0,1\})^*$ as the language of words of $L$ for which exactly one letter between position $\lfloor\alpha n\rfloor$ and $\lfloor\beta n\rfloor$ has been marked ($n$ is the length of the word).
	
	Formally, the word $(u_i,\delta_i)_{i\in\llbracket 0, n-1\rrbracket} \in (\Sigma\times\{0,1\})^*$ is in $\Lprop{\alpha,\beta}$ if and only if $u_0u_1\ldots u_{n-1}\in L$, there is exactly one $i_0$ such that $\delta_{i_0} = 1$ and $i_0\in\llbracket\lfloor\alpha n\rfloor, \lfloor\beta n\rfloor\rrbracket$.
\end{definition}

The rest of this section will be devoted to the proof of the following theorem:
\begin{theorem}
	\label{theo:markers}
	For any language $L$ and any $\alpha, \beta\in [0,1]$ with $\alpha < \beta$,
	\begin{displaymath}
		\Lprop{\alpha,\beta} \in \CA(n) \Rightarrow L \in \CA(n)
	\end{displaymath}
\end{theorem}

First, notice that it is sufficient to prove the theorem for $\alpha$ and $\beta$ rationals with $0<\alpha<\beta<1$. Indeed, for any $\alpha, \beta \in [0,1]$ and $\alpha', \beta'\in\QQ$ such that $\alpha \leq \alpha' < \beta' \leq \beta$ if $\Lprop{\alpha,\beta}\in \CA(n)$ then we can recognize $\Lprop{\alpha',\beta'}$ in real time by simulating the automaton $\ACA$ that recognizes $\Lprop{\alpha,\beta}$ in real time while simultaneously verifying that the marker is placed between $\lfloor\alpha' n\rfloor$ and $\lfloor\beta' n\rfloor$, which can be done in real time because $\alpha'$ and $\beta'$ are both rationals. If the marker is in the right range then the input word is in $\Lprop{\alpha',\beta'}$ if and only if $\ACA$ accepts it. If the marker is not in the right range, then the word is not accepted. From now on, we can therefore assume that $\alpha, \beta\in\QQ$ and $0<\alpha<\beta<1$.

To prove the theorem, we assume that we have a CA $\ACA$ that recognizes the language $\Lprop{\alpha,\beta}$ in real time, for some $\alpha, \beta\in [0,1]$, with $\alpha < \beta$. We will show how to make a CA $\ACA'$ that recognizes $L$ in real time.

The construction will be done in two steps. First we describe a CA that starts with a fixed set of positions marked (independently of the input length) and we show that with these markers we can recognize $L$ in real time. Then we transform this CA into one that does not need the positions to be marked ahead of time.

\subsection{Universal Markers}

Let $n_0$ be the smallest integer such that $1 + \frac{1}{2^{n_0}} < \frac{\beta}{\alpha}$ and consider the set $M$ of integers whose binary representation is such that all the digits 1 are on the $(n_0+1)$ most significant bits (see Figure~\ref{fig:M example}):
\begin{displaymath}
	M = \{ x\times 2^k \ |\ x \in\ \rrbracket 0, 2^{n_0+1}-1\rrbracket, k \in \NN\} = \llbracket 0, 2^{n_0}-1\rrbracket \cup \{ x\times 2^k \ |\ x \in \llbracket 2^{n_0}, 2^{n_0+1}-1\rrbracket, k\in\NN\}
\end{displaymath}

The set $M$ contains an initial segment $\llbracket 0, 2^{n_0}\llbracket$ and copies of $\llbracket 2^{n_0}, 2^{n_0+1}\llbracket$ multiplied by the powers of $2$ (indicated as bracketed ``blocks'' in Figure~\ref{fig:M example}).

\begin{figure*}[htbp]
	\centering
	\includegraphics[scale=.8]{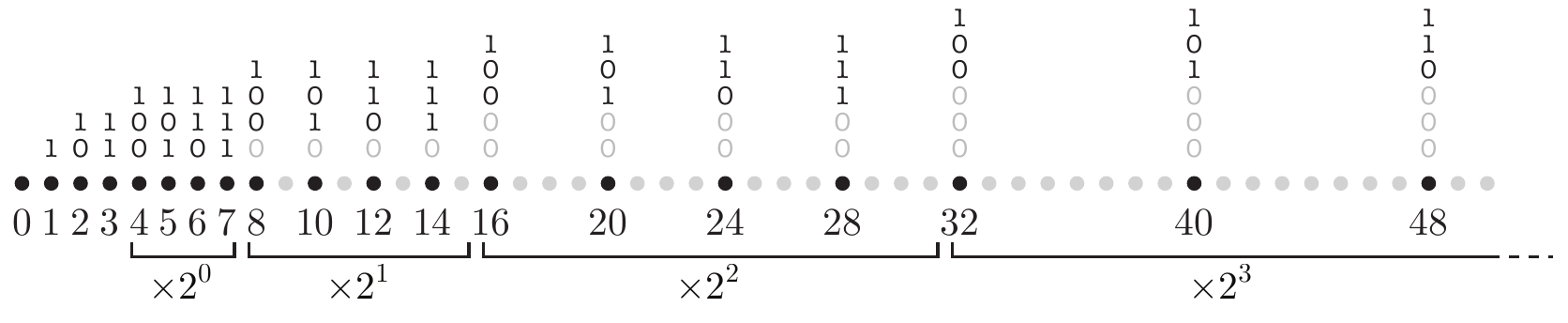}
	\caption{The set $M$ for $n_0=2$.}
	\label{fig:M example}
\end{figure*}

Let us now consider the ratio between consecutive elements of $M$. Denote by $(m_i)_{i\in\NN}$ the elements of $M$ in increasing order. For all $m_i\geq 2^{n_0}$, we have
\begin{equation}
	\label{eqn:bounds mi}
	1 + \frac{1}{2^{n_0+1}-1} \leq \frac{m_{i+1}}{m_i} \leq 1 + \frac{1}{2^{n_0}}
\end{equation}

The lower bound corresponds to the ratio between the last element of a block and the first of the next block (in the example with $n_0=2$, this ratio is $\frac{8}{7}$) and the upper bound corresponds to the ratio between the two first elements of a block (in the example it is $\frac{5}{4}$).

From the definition of $n_0$, we have
$\forall m_i \geq 2^{n_0}, \frac{m_{i+1}}{m_i} < \frac{\beta}{\alpha}$
and
$\frac{m_{i+1}}{\beta} < \frac{m_i}{\alpha}$
so intervals
$[\frac{m_i}{\beta}, \frac{m_i}{\alpha}]$
and
$[\frac{m_{i+1}}{\beta}, \frac{m_{i+1}}{\alpha}]$
overlap, and hence
$[2^{n_0}, +\infty[\ \subseteq\ \bigcup_{m\in M} \left[\frac{m}{\beta}, \frac{m}{\alpha}\right]$.

From this, we get $\forall x\geq 2^{n_0}, \exists m\in M, \lfloor\alpha x\rfloor \leq m \leq \lfloor\beta x\rfloor$.

Since $M$ contains all the elements in the the missing initial segment $\llbracket 0, 2^{n_0}\rrbracket$, we have proved the following lemma:
\begin{lemma}
	\label{lem:interval}
	$\forall n\in\NN, \exists m\in M, \quad \lfloor\alpha n\rfloor \leq m\leq \lfloor\beta n\rfloor$
\end{lemma}

We now know that for any input word $w$ of length $n$, at least one of the elements of $M$ lies between $\lfloor\alpha n\rfloor$ and $\lfloor\beta n\rfloor$ and can therefore be used by $\ACA$ as a marker to know whether $w$ is in $L$.

Moreover, from Equation~(\ref{eqn:bounds mi}), we have $\forall i,k\in \NN,$
\begin{equation}
	\label{eqn:growth}
	m_i\left(1 + \frac{1}{2^{n_0+1}-1}\right)^k \leq m_{i+k}
\end{equation}

Define $k_0$ as the smallest integer such that
\begin{equation}
	\label{eqn:defi k0}
	\frac{1}{\alpha} \leq \left(1 + \frac{1}{2^{n_0+1}-1}\right)^{k_0}
\end{equation}

Equations~(\ref{eqn:growth}) and (\ref{eqn:defi k0}) state that for any $i$, if $m_{i+k_0+1}\leq n$ then $m_{i+1} \leq \lfloor\alpha n\rfloor$ and thus $m_{i} < \lfloor\alpha n\rfloor$, which means that if a word is long enough to have a letter on $m_{i+k_0+1}$, then $m_i$ is out of the range of valid markers. As a consequence, it is never necessary to consider more than $(k_0+1)$ elements of $M$ at any given time.

We can now describe the first part of the construction. We assume that the automaton is given as input a word $w$ of $\Sigma^*$ on which all letters at indexes in $M$ are marked. On such an input the automaton simulates the behavior of $\ACA$ on $w$ (as if no letter was marked), but each marked cell also starts a separate simulation of $\ACA$ that considers that the letter is the only marked letter of the input. However, as previously observed, only the $(k_0+1)$ simulations corresponding to the largest elements of $M$ are significant, all others correspond to markers at positions before the required range. This means that each cell only needs to simulate at most $(k_0+1)$ computations of $\ACA$ and whenever a new computation should be taken into account, the one corresponding to the lowest element of $m$ is discarded.

\begin{figure*}[htbp]
	\centering
	\includegraphics[scale=1]{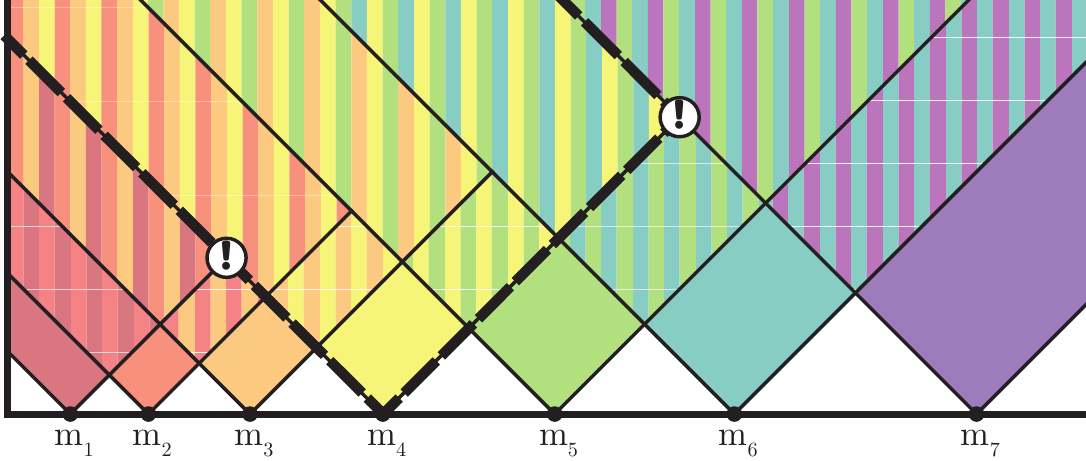}
	\caption{Space-time diagram of the simulation. Each cone represents the area on which a simulation corresponding to an element of $M$ is done. The cones extend to the left until the origin, but are interrupted to the right when there are too many simulations running at once.}
	\label{fig:cones}
\end{figure*}

This behavior is illustrated by Figure~\ref{fig:cones}. This figure represents a space-time diagram of the automaton. Each cone corresponds to a simulation of $\ACA$ for which the marked cell is the origin of the cone. The figure corresponds to a case where $k_0=2$, meaning that at most 3 simulations are performed in parallel by each cell. The thick dashed line illustrates the area of the space-time diagram on which the simulation corresponding to the marker on $m_4$ is performed. Since this specific simulation starts on the cell $m_4$, all space-time sites outside of the cone starting from that cell are not performing this specific simulation (they are however simulating the behavior of $\ACA$ without any marker, wich coincides with the behavior with a marker on $m_4$ on said sites out of the cone). As time passes, more and more cells are included in this cone and start performing this specific simulation. The two sites indicated by \stB{} correspond to events where a cell enters a fourth cone. Instead of starting a fourth simulation, it discards the simulation corresponding to the lowest $m_i$: on the left, the simulation for a marker at $m_1$ is discarded, on the right it's the simulation corresponding to the marker at $m_4$ that is discontinued.

In each simulation of $\ACA$, the automaton checks that the marker is located between $\lfloor\alpha n\rfloor$ and $\lfloor\beta n\rfloor$ by sending two signals from the marker towards the origin, one at speed $\alpha$ and the other at speed $\beta$ (it is possible if $\alpha$ and $\beta$ are rationals). If the marker is in the correct range, the signal moving at speed $\beta$ will arrive before real time while the one moving at speed $\alpha$ will arrive after real time.

From Lemma~\ref{lem:interval}, we know that there is a marker $m$ in the correct range and since $\ACA$ properly recognizes $\Lprop{\alpha,\beta}$, the simulation of $\ACA$ for the marker $m$ will let the automaton know whether $w$ is in $L$ or not.

\subsection{Construction of $M$}
\label{ssec:construction_of_M}

We now have to remove the requirement that the elements of $M$ be marked on the input. To do this, we use a space-time compression technique: instead of starting the computation immediately, the automaton moves the states from the input word towards the orgin to group them three by three, and only then simulates the behavior of the original (uncompressed) CA. By performing such a compression, the initial configuration is mapped to the space-time line of slope 2, and the computation takes place in the cone between this line and the vertical axis (as illustrated by the green cone in Figure~\ref{fig:marking} in which each dark green cell on the right border holds 3 states from the initial configuration). Although the space-time diagram is strongly modified by the compression, the computation of the states on the origin cell does not suffer any slow down.

To mark the elements of $M$ for the compressed computation, the CA builds a binary counter on the main diagonal of the space-time diagram to obtain the index of each transverse diagonal (see Figures~\ref{fig:marking} and \ref{fig:marking_detail}).

\begin{figure}[t]
\begin{minipage}[b]{.475\linewidth}
	\centering
		\includegraphics[page=1]{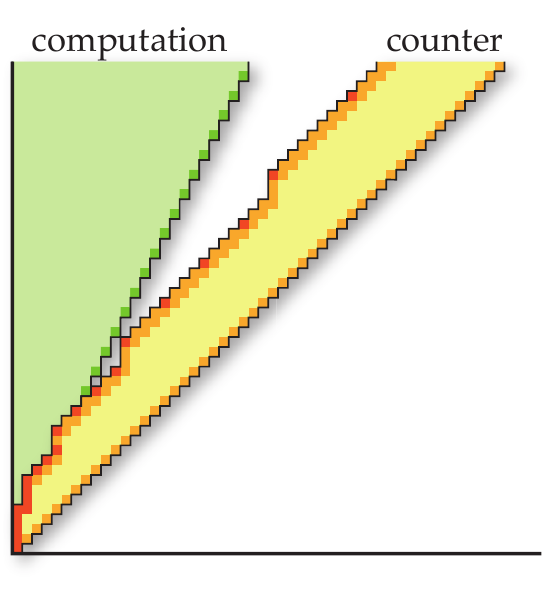}
	\caption{The diagonal counter is added on the compressed space-time diagram. Because the diagonal indexes are of logarithmic length, the counter and the main computation overlap only on a finite number of cells. Diagonals of indexes in $M$ (here with $n_0=2$) are represented with a darker square.}
	\label{fig:marking}
\end{minipage}%
\hspace{.05\linewidth}%
\begin{minipage}[b]{.475\linewidth}
	\centering
		\includegraphics[page=2]{marking.pdf}
	\caption{Detailed behavior of the diagonal counter. Least significant bits are on the bottom right. For better readability the digits of indexes on odd diagonals are represented in light grey. Indexes in $M$ are marked with a (*) (here M is defined with $n_0=2$) and the signal sent towards the main computation is represented by an arrow.}
	\label{fig:marking_detail}
\end{minipage}
\end{figure}

Since it is easy to recognize binary representations of elements in $M$ (all bits but the $(n_0+1)$ most significant must be $0$), a signal can be sent along the diagonals whose index is in $M$ so that the letters of the input word at positions in $M$ can be marked before the compressed computation effectively starts.

From an input word $w$ in $\Sigma^*$, the automaton can therefore simulate the previously described automaton as if the positions in $M$ were marked from the start, and determine in real time whether $w$ is in $L$. This concludes the proof of Theorem~\ref{theo:markers}.

\section{Central Compression}
\label{sec:central compression}

In this section we describe a way to simulate the behavior of a 1-dimensional CA $\ACA$ on an input $w\in\Sigma^*$ working in real time with another 1-dimensional CA $\ACA'$ on input $w$ with a marked position, by compressing the space-time diagram of $\ACA$. 

\subsection{General Description}

Assume that a special position has been marked on the input word of $\ACA'$. We want to group the states of the original simulated CA by groups of 3 around the mark as illustrated by Figure~\ref{fig:compressed sim}. To do so, the letters of the input word (represented as large circles in the figure) are shifted towards the marked position (indicated by a thick dashed line). Because the letters do not know in advance whether the mark is to their right or to their left, the AC uses two separate layers, one that shifts the letters to the right and the other to the left. Letters that move away from the mark will never be grouped and will not affect the simulation.

\begin{figure}[htb]
	\begin{minipage}{.475\linewidth}
		\centering
		\includegraphics[page=5, width=.7\linewidth]{csim.pdf}
		\caption{Simulation of a CA by compressing its space-time diagram by a factor 3 around a given mark (indicated by a thick dashed line). Each cell in the simulating area holds 3 states that correspond to states in the original space-time diagram. Numbers on the cells indicate which time step of the original automaton they are currently simulating.}
		\label{fig:compressed sim}
	\end{minipage}%
	\hspace{.05\linewidth}%
	\begin{minipage}{.475\linewidth}
		\centering
		\includegraphics{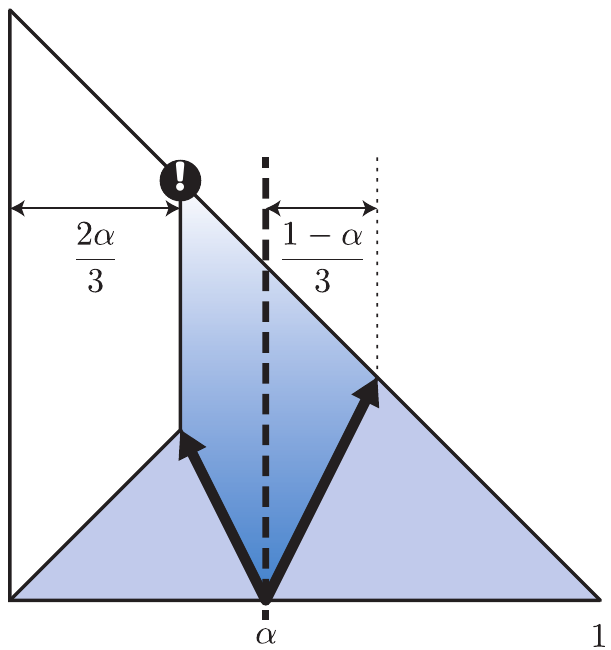}
		\caption{Diagram of the compression around a mark (thick dashed line) at position $\lfloor\alpha n\rfloor$. The thick arrows illustrate the sites where the letter of the input words are first taken into account. The site where the result of the simulated automaton is computed is indicated by the exclamation mark.}
		\label{fig:compressed schema}
	\end{minipage}
\end{figure}

When letters reach the marked position, they stack on the corresponding cell. When a cell has 3 letters, it is considered full and its neighbors start gathering letters in turn. In Figure~\ref{fig:compressed sim}, grouped states are represented by small circles.

Once a cell is fully grouped, it watches its neighbors until it has enough information to simulate 3 steps of the original automaton at once on all its grouped states. This happens when its neighbors are fully grouped and their simulated time is at least equal to its own. During the compressed simulation, the difference between the simulated times of a cell and its neighbor is at most 3 (it can be -3, 0 or 3 since the simulation advances by 3 steps at a time). If a cell advances faster than its neighbor it has to memorize its previous state so that the neighbor can use it when doing its own transition. In Figure~\ref{fig:compressed sim}, the number in each cell represents the simulated time step: a cell numbered 3 for instance has to wait until both its neighbors are labelled 3 or more before it can compute step 6 for all 3 of its grouped states. Cells represented with a grey square are cells that don't contain any significant information (they correspond to sites that are outside of the real time cone in the original space-time diagram) so their neighbors do not need to wait for their information.

If the mark around which the cells are grouped is in the first half of the input word (which is the case in Figure~\ref{fig:compressed sim} as there are 9 letters left of the mark and 15 right) the simulation can take place properly as the left part can compute its states faster and have the information ready for the right part. The key point is that from the time when the rightmost cell is fully grouped (this time is indicated by a thick horizontal line in Figure~\ref{fig:compressed sim}) all the sites on the diagonal must be able to advance their computation by 3 steps. This guarantees that the leftmost cell has eventually computed as many steps of the original diagram as if it had started at the indicated time and advanced by 3 steps each time, which corresponds to the whole computation of the original diagram.

In Figure~\ref{fig:compressed sim}, the leftmost cell of the compressed area seems to be 2 steps behind real time at the end of the simulation (the initial configuration has 24 letters so the real time is 23). However, because the cell holds the states of the 3 leftmost cells of the original configuration at time 21, it has all the relevant information to determine the state of the origin at time 23.

There are of course many rounding problems when the number of letters left or right or the mark is not a multiple of 3. However these roundings cause at most a constant delay, which can be corrected by using a constant speed-up theorem \cite{mazoyer92}.

\subsection{Properties}

By compressing the space-time diagram of the automaton around the marked position, we are able to perform the same computation with some significant differences.

First, the letters of the input word are not taken into account from the start of the computation, but rather in a sequential order. The time at which a given letter of the input word is effectively considered to determine the result of a transition in the simulation is proportional to its distance to the initial mark (see Figure~\ref{fig:compressed schema} in which the sites where the letters of the input word are first taken into account are along the thick arrows).

Second, the result of the computation is obtained significantly before real time, on a cell that it not the origin. The site where the result is obtained is represented by \stE{} in Figure~\ref{fig:compressed schema}.

Let us denote by $\alpha$ the proportion of the word at which the mark is set. When using the compression in a later section, we will need the free space left of the compressed area (which is of width $\frac{2\alpha}{3}$) to be larger than each of the sides of the compressed area. As said before, we need $\alpha \leq \frac 1 2$ for the simulation to work without delay, which means that the left side of the compressed area is smaller than the right side. Since the right side is of width $\frac{1-\alpha}{3}$, this means that we want $\frac{1-\alpha}{3} \leq \frac{2\alpha}{3}$, which amounts to $\frac 1 3 \leq \alpha \leq \frac 1 2$.

\section{The Power of Space}
\label{sec:power of space}

In this section we compare the computational power of 2-dimensional CA working on the Moore neighborhood to that of 1-dimensional CA working on the standard neighborhood.

\begin{definition}
	Given a marked language $L\subseteq (\Sigma\times\{0,1\})^*$, we define $\widetilde L\subseteq \Sigma^*$ as the language obtained by removing the marks of words in $L$ ($\widetilde L$ can be seen as the result of the first projection map on $L$).
\end{definition}

\begin{theorem}
	\label{theo:linear single mark}
	For any language $L\subseteq (\Sigma\times\{0, 1\})^*$ of words having at most one marked position, $L\in \CA(2n) \Rightarrow \widetilde L\in\CA_2(2n)$.
\end{theorem}
\begin{proof}
	Let $L\subseteq (\Sigma\times\{0, 1\})^*$ be a language in $\CA(2n)$ of words having at most one marked position and $\ACA$ be a 1-dimensional CA that recognizes $L$ in time $n\mapsto 2n$. Let us describe a 2-dimensional CA $\ACA'$ that recognizes $\widetilde L$ in linear time.
	
	The input of $\ACA'$ is an unmarked word $w\in\Sigma^*$ of length $n$. The idea is to use the second dimension of $\ACA'$ to run $n$ simulations of $\ACA$, one for each possible position of the mark as shown on Figure~\ref{fig:linearsim}.
	
\begin{figure}[htb]
	\begin{minipage}{.475\linewidth}
		\centering
		\includegraphics[width=\linewidth]{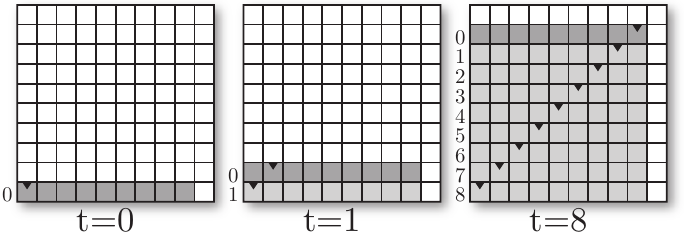}
		\caption{Running $n$ simulations of a 1-dimensional CA with a 2-dimensional CA.}
		\label{fig:linearsim}
	\end{minipage}%
	\hspace{.05\linewidth}%
	\begin{minipage}{.475\linewidth}
		\centering
		\includegraphics[width=\linewidth,page=2]{linearsim.pdf}
		\caption{Real time parallel simulations of a 1DCA. At each time, a new line copies the current unmarked simulation of $\ACA$ and continues from the same point (no delay). Each line has up to two special positions that are marked when the line starts simulating $\ACA$.}
		\label{fig:special positions}
	\end{minipage}
\end{figure}

	At time $t=0$, the input is on the first line and a simulation of $\ACA$ starts on that line with a mark on the leftmost cell of the word. At each subsequent time, the original input is copied on the next line (moving up), and a new simulation of $\ACA$ is started on that line with the mark on the next position (moving right). At time $t=n$, the first line has simulated $n$ steps of $\ACA$ with a mark on the first cell, while the $n$-th line starts a new simulation with the mark on the last cell.
	
	When a simulation finishes, the result is sent back towards the origin (on the first line). At time $t=3n$ the simulation on line $n$ is finished, and at time $t=4n$ the results of all simulations are available on the origin (to be complete, the automaton must also run an extra simulation on the first line that simulates the behavior of $\ACA$ on input $w$ without any mark). The language $\widetilde L$ is therefore recognized in time $n\mapsto 4n$, and by using a linear acceleration we get $\widetilde L\in\CA_2(2n)$.
\end{proof}

\begin{theorem}
	\label{theo:single mark}
	For any language $L\subseteq (\Sigma\times\{0, 1\})^*$ of words having at most one marked position, $L\in \CA(n) \Rightarrow \widetilde L\in\CA_2(n)$.
\end{theorem}
\begin{proof}
	The basic idea is again to run parallel simulations of a 1-dimensional CA on the lines of the 2-dimensional CA, but because the automaton must work in real time it is not possible to waste a linear time starting the simulations, nor a linear time bringing back the results of the farthest simulation to the origin, which is why we use the compressed simulation presented in Section~\ref{sec:central compression}.
	
	Consider a language $L\in (\Sigma\times\{0,1\})^*$ such that all words of $L$ have at most one marked position and a 1-dimensional CA $\ACA$ that recognizes $L$ in real time. We want to describe a 2-dimensional CA $\ACA_2$ that takes an unmarked word $w\in \Sigma^*$ of length $n$ and decides in real time if by adding at most one mark to $w$ we can obtain a word in $L$.
	
	For now, let us assume that the input word $w$ has a mark on a position between $\lfloor\frac{n}{3}\rfloor$ and $\lfloor\frac{n}{2}\rfloor$ so that we can run the compressed simulation easily. This mark will be referred to as the \emph{compression mark}, it is different from the marks of $L$ that we want to simulate on each line.
	
	Instead of simulating $\ACA$, $\ACA_2$ will simulate an automaton $\ACA_C$ that simulates $\ACA$ with a central compression on the compression mark as described in Section~\ref{sec:central compression}.
	
	The behavior of $\ACA_2$ is as follows:
	\begin{itemize}
		\item at time $t=0$, the first line of $\ACA_2$ starts a simulation of $\ACA_C$ as if no letter of the input was marked;
		\item at each time, this ``unmarked'' simulation is copied to the next line (moving up), each line continues the simulation from the step at which it is when copied (so that more and more lines are performing the same simulation of $\ACA_C$, without suffering any delay);
		\item meanwhile, each line has one or two \emph{special positions}. The special position of the first line is the one where the compression mark is, and the special positions of line $(i+1)$ are the one left of the leftmost special position of $i$, and the one right of the rightmost position of $i$ (see Figure~\ref{fig:special positions}). Special positions on each line are marked when the line copies the simulation from the previous line.
		\item during the simulation of $\ACA_C$ by a line, when one of the cells at a special position finishes grouping 3 letters of the input word, 3 new simulations of $\ACA_C$ are started on this line, each considering that there was a mark on one of the letters of the input word that was grouped on the special position. Because each line has at most two special positions, at most 7 simulations of $\ACA_C$ are run in parallel on each line (3 for each special position and the unmarked one).
	\end{itemize}
	
	This construction works because of the properties of the central compression discussed in Section~\ref{sec:central compression}.
	
	First, the simulations of $\ACA_C$ on each line can be performed properly without any delay because the time at which an input letter that eventually is grouped on a special position of the line becomes significant is after the activation of the line and the marking of its special positions. Therefore, all simulations of $\ACA_C$ on all lines are synchronized, and the farther lines do not suffer any delay.
	
	Second, the number of lines really used (on which significant simulations that correspond to a potential mark on an input letter) is equal to the length of the largest side of the compressed area (which is $\frac{1-\alpha}{3}$, see Figure~\ref{fig:compressed schema}). This length is less than the time remaining when the simulations of $\ACA_C$ obtain their result ($\frac{2\alpha}{3}$), so it means that there is enough time to send back the result of the parallel simulations to the origin in real time.
	
	The last detail is now to remove the requirement for the compression mark to be given as input. From Theorem~\ref{theo:markers}, we know that if the computation can be performed in real time with a mark anywhere between positions $\lfloor\frac{n}{3}\rfloor$ and $\lfloor\frac{n}{2}\rfloor$ then it can be done in real time without previous marking. Technically, Theorem~\ref{theo:markers} only applies to 1-dimensional CA, but in this case the 2-dimensional CA performs 1-dimensional computations on each line almost independently so by having each line perform the construction from the proof of Theorem~\ref{theo:markers} we obtain the result for this specific 2-dimensional CA.
\end{proof}

\section{Consequences}
\label{sec:consequences}

Let us now discuss some consequences of Theorem~\ref{theo:single mark}.

\begin{corollary}
	The concatenation $L_1L_2$ of two 1-dimensional real time languages $L_1$ and $L_2$ is recognizable in real time by a 2-dimensional CA working on the Moore neighborhood.
\end{corollary}
\begin{proof}
	Given a word $w=uv$ with a mark between $u$ and $v$, it is easy to check in real time if $u\in L_1$ and $v\in L_2$, so from Theorem~\ref{theo:single mark} the unmarked language is in $\CA_2(n)$.
\end{proof}

\begin{corollary}
	If $\CA(n) =\CA_2(n)$, $\CA(n)$ is closed under concatenation.
\end{corollary}

Without the assumption that $\CA(n)=\CA_2(n)$, it is still unknown whether $\CA(n)$ is closed under concatenation. Actually, it is also unknown whether $\CA(2n)$ is closed under concatenation and even if the concatenation of two languages in $\CA(n)$ is in $\CA(2n)$.

\begin{corollary}
	For any language $L$ and any $\alpha \in \QQ \cap [0,1],\ \Lprop{\alpha} \in \CA(n) \Rightarrow L \in \CA_2(n)$.
\end{corollary}
\begin{proof}
	Given a word with one marked position, it is easy to check simultaneously in real time if the mark is at position $\lfloor\alpha n\rfloor$ and if the marked word is in $\Lprop{\alpha}$.
\end{proof}

\begin{corollary}
	If $\CA(n)=\CA_2(n)$, for any language $L$ and any $\alpha \in \QQ \cap [0,1],\ \Lprop{\alpha} \in \CA(n) \Rightarrow L \in \CA(n)$.
\end{corollary}

Under the assumption that $\CA(n)=\CA_2(n)$, Theorem~\ref{theo:single mark} can also be strengthened:

\begin{corollary}
	\label{cor:poly}
	If $\CA(n)=\CA_2(n)$, for any $k\in\NN$ and any language $L\subseteq (\Sigma\times\{0, 1\})^*$ of words having up to $k$ marked positions, $L\in \CA(n) \Rightarrow \widetilde L\in\CA(n)$.
\end{corollary}
\begin{proof}
	By induction on $k$. The case $k=1$ is a direct consequence of Theorem~\ref{theo:single mark} and the assumption that $\CA(n)=\CA_2(n)$.
	
	If the corollary is true for up to $k$ marks, and $L$ is a language of words with up to $(k+1)$ marks, let us consider the language $L'\subseteq (\Sigma\times\{0,1\}^2)^*$ obtained from $L$ by changing the first marked letter $(u_i,1)$ to $(u_i,0,1)$, all other marked letters $(u_i,1)$ into $(u_i,1,0)$ and all unmarked letters $(u_i,0)$ into $(u_i,0,0)$ (effectively distinguishing the first mark from the others).
	
	If $L\in\CA(n)$ then $L'\in\CA(n)$ since it's possible to simulate the recognition of $L$ by considering that all marks are the same, while independently checking in real time that the distinguished mark is the first marked position. From Theorem~\ref{theo:single mark} the language of words in $L$ in which the first mark has been removed is therefore in $\CA_2(n)$ and hence also in $\CA(n)$. The words of this language have at most $k$ marks so from the induction hypothesis, $L\in \CA(n)$.
\end{proof}

\newpage
\bibliographystyle{plain}
\bibliography{stacs27Grandjean}

\end{document}